\documentclass[aps,floatfix,onecolumn,a4paper,noshowpacs, nofootinbib,superscriptaddress,11pt]{revtex4}
\usepackage{float,xcolor,upgreek,yfonts,tikz}
\usepackage{color}
\usepackage{xcolor}
\usepackage{bm}
\usepackage{comment}
\usepackage{graphicx,bigints}
\usepackage{amsthm}
\usepackage{graphicx,float}\usepackage[all]{xy}
\usepackage{amsmath,amssymb,upgreek}

   \usepackage{caption}

\usepackage{subcaption}
\usepackage{physics}
\usepackage{tensor}
\usepackage{amssymb}

\usepackage{alphabeta}
 \usepackage{cancel}
\usepackage{dcolumn}
\usepackage{textgreek}
\usepackage{adjustbox}
\usepackage{comment}
\usepackage{multirow} 
\usepackage{amsmath}
\allowdisplaybreaks
\usepackage{units}
\usepackage{overpic}

\newcommand{\beq}{\begin{eqnarray}}
\newcommand{\eeq}{\end{eqnarray}}
\newcommand{\be}{\begin{equation}}
\newcommand{\ee}{\end{equation}}

\newcommand{\bea}{\begin{eqnarray}}
\newcommand{\eea}{\end{eqnarray}}

\newcommand{\ba}{\begin{eqnarray}}
\newcommand{\ea}{\end{eqnarray}}

\usepackage{framed}
\usepackage{tikz}
\usepackage{tikz-cd}
\usepackage{amsmath, amsfonts, amssymb, amsthm, mathtools}
\newtheorem{thmTEMA}{Theorem}[section]
\newtheorem{coroTEMAi}{Corollary}[section]
\newtheorem{defTEMAi}{Definition}[section]

\usepackage[colorlinks,hyperindex,unicode]{hyperref}
\definecolor{green1}{RGB}{0,128,0} 
\hypersetup{hidelinks,backref=true,pagebackref=true,hyperindex=true,colorlinks=true,breaklinks=true,urlcolor= blue}
\hypersetup{%
  colorlinks = true,
  linkcolor  = blue,
  citecolor = cyan,
}
\usepackage{bookmark,textgreek}
\usepackage{hyperref,color,xcolor}
\hypersetup{hidelinks,hyperindex=true,colorlinks=true,breaklinks=true,urlcolor= blue}
\hypersetup{%
  colorlinks = true,
  linkcolor  = blue
}
\newcommand\orcidjulio{{\href{https://orcid.org/0000-0001-7405-8852}{\orcidicon}}}

\newcommand\orcidrogerio{{\href{https://orcid.org/0000-0001-7848-5472}{\orcidicon}}}

\newcommand{\orcidicon}{%
	\begin{tikzpicture}
	\draw[lime, fill=lime] (0,0)
		circle [radius=0.16]
		node[white] {{\fontfamily{qag}\selectfont \tiny ID}};
	\draw[white, fill=white] (-0.0625,0.095)
		circle [radius=0.007];
	\end{tikzpicture}	\hspace{-2mm}
}

\usepackage{graphicx}
\usepackage{dcolumn}
\usepackage{bm}
\usepackage{multirow}
\usepackage{caption}
\usepackage{subcaption}
\usepackage{comment} 
\usepackage{float}

\begin{document}

\title{A unified approach to spinor duals via Clifford algebras and $G_\Omega$ groups}

\author{R. T. Cavalcanti\orcidrogerio}
\email{rogerio.cavalcanti@ime.uerj.br}
\affiliation{Institute of Mathematics and Statistics, Rio de Janeiro State University, 20550-900, Rio de Janeiro, Brazil}
\affiliation{Department of Physics, São Paulo State University, 12516-410, Guaratingueta, Brazil}

\author{J. M. Hoff da Silva\orcidjulio}
\email{julio.hoff@unesp.br}
\affiliation{Department of Physics, São Paulo State University, 12516-410, Guaratingueta, Brazil}

\medbreak
\begin{abstract} 
Recent developments in the construction of generalized Dirac duals have revealed, within the structure of the Clifford algebra $\mathbb{C}\otimes\mathcal{C}\ell_{1,3},$ the existence of distinct algebraic formulations of spinors duals with potential applications in quantum field theoretic models.  In this work, after reviewing the matrix formulation, we employ the recent covariant formulation of the generalized spinor dual and establish its interplay with the algebra $\mathcal{C}\ell_{1,3}$. We construct dual mappings governed by groups denoted by $G_\Omega$ and introduce the notion of $\Omega$-equivalence classes as a tool to classify dual spinors from a group-theoretic perspective.
\end{abstract}



\maketitle

\section{Introduction}

Inner products and related structures, such as quadratic forms, bilinear forms, and metrics, constitute a set of fundamental and essential mathematical ingredients in modern physics. In quantum mechanics, the inner product on Hilbert spaces plays a central role, enabling the extraction of the theory's observables \cite{sakurai1986modern}. In special relativity, the bilinear form known as the Minkowski metric facilitated a far more powerful and elegant reformulation of the theory, while also paving the way for general relativity, where metrics are defined on more general Lorentzian manifolds \cite{wald2024general}. Dirac theory is no exception, as the inner product is also central. The key difference, however, is that in these other theories, the relevant inner products or metrics are either fixed and determined by the theory's structure, as in quantum mechanics and Minkowski spacetime, or determined \emph{a posteriori}, as in general relativity, through the Einstein field equations. In Dirac theory, by contrast, the inner product is a choice—a canonical one, since it yields the fundamental results of quantum field theory in particle physics, such as the appropriate orthogonality relations and locality properties for fermionic fields \cite{Peskin:1995ev}—yet it remains, ultimately, a choice. Despite its undeniable importance, this choice may obscure potential avenues for investigating foundational problems in theoretical physics, such as the dark matter problem and the search for first-principle candidates to describe it \cite{Ahluwalia:2004sz, Ahluwalia:2004ab, Ahluwalia:2016rwl}, or even as a source for controlled \emph{ab initio} Lorentz symmetry-breaking theories. A consistent investigation into the reformulation and generalization of the dual originally proposed by Dirac must therefore revisit the foundational aspects of spinor theory \cite{Vaz:2016qyw}.

Spinors have played a fundamental role in the physics of fermionic particles since Pauli’s formulation of spin \cite{pauli} and Dirac’s theory of the electron \cite{dirac}. The underlying mathematical concept, however, originates in the classical work of Élie Cartan \cite{cartan}, with a subsequent algebraic version systematized by Chevalley \cite{chev}. Since then, spinor theory has undergone extensive development, and both its mathematical and physical aspects have acquired great relevance in high-energy physics \cite{Peskin:1995ev}, geometry \cite{Lawson:1998yr}, and the study of torsion in gravity \cite{Kibble:1961ba, Hehl:1976kj, Vignolo:2021frk, Fabbri:2024avj}. As mentioned, novel duals have emerged in quantum theories proposed as candidates to describe dark matter \cite{Ahluwalia:2004ab,Ahluwalia:2009rh, Agarwal:2014oaa}. In particular, Elko spinors \cite{Ahluwalia:2004ab}, whose most relevant feature for our present analysis is that they require a new dual formulation to achieve physical consistency \cite{Ahluwalia:2004ab}. This necessity culminated in the development of a systematic theory for spinor duals \cite{Ahluwalia:2016rwl,epls, Rogerio:2017hwz, HoffdaSilva:2019ykt, deGracia:2024yei}, based on a careful set of physical and formal requirements. The possibility of defining new duals has also opened the way to revisiting Lounesto’s classification of spinors \cite{Lounesto:2001zz, Cavalcanti:2014wia, HoffdaSilva:2017waf} and its subsequent extensions \cite{daRocha:2005ti, Bonora:2014dfa, Fabbri:2017lvu, Arcodia:2019flm, BuenoRogerio:2019kgd, BuenoRogerio:2019yqb, Rogerio:2023kcp, Rogerio:2024lfh, Rogerio:2025avn}.

In this paper, we continue the investigation initiated in \cite{Ahluwalia:2016rwl} and further extended and generalized in \cite{Rogerio:2017hwz, HoffdaSilva:2019ykt, Cavalcanti:2020obq}, focusing on generalized spinor duals and the use of algebraic tools in their analysis. In particular, we employ the covariant form of the generalized spinor dual proposed in \cite{Rogerio:2025avn} and its relation to Clifford algebras to explore possible dual mappings, as introduced in \cite{Cavalcanti:2020obq}, along with their associated groups. We argue that these groups can be used to define a spinor classification based on group-theoretic results. The paper is organized as follows: In Section \ref{formulacao}, we review basic results on Clifford algebras, algebraic spinor spaces, generalized dual structures, and their matrix and multivector formulations. In Section \ref{algebras}, we revisit the proposal of studying spinor duals via other algebraic structures, with a focus on the use of Lie groups. We introduce the concept of $\Omega$-equivalence classes and their application in classifying dual structures, concluding the section with specific examples of groups. In Section \ref{cl13}, we connect the results from the previous sections and investigate the groups defined within the algebra associated with the generalized dual, namely $\mathcal{C}\ell_{1,3}$. Finally, in Section \ref{conclusao}, we summarize and discuss the main results obtained.


\section{Algebraic Formulation and General Dual Spinors}\label{formulacao}

The formal structure of spinors is most clearly revealed when studied within the framework of Clifford algebras~\cite{Vaz:2016qyw, Lounesto:2001zz}, whose definition and key results, pertinent to our discussion, are summarized below. Given a real vector space endowed with a symmetric bilinear form $g: \mathbb{R}^n \times \mathbb{R}^n \to \mathbb{R}$ of signature $(p,q)$, denoted by\footnote{$n=p+q$.} $\mathbb{R}^{p,q}$, its associated Clifford algebra is defined as follows:
\begin{defTEMAi}
The Clifford algebra $\mathcal{C}\ell_{p,q}$ associated with the quadratic space $\mathbb{R}^{p,q}$ is the unital associative algebra such that:
\begin{enumerate}
\item[(i)]\label{it1} The Clifford map $\gamma:\mathbb{R}^{p,q} \to \mathcal{C}\ell_{p,q}$ is linear and satisfies
$$
\gamma(v)\gamma(u) + \gamma(u)\gamma(v) = 2\,g(v,u),
\qquad \forall\, v,u \in \mathbb{R}^{p,q};
$$
\item[(ii)] If $(\mathcal{Y},\gamma')$ is another unital associative algebra and $\gamma' : \mathbb{R}^{p,q} \to \mathcal{Y}$ satisfies
$$
\gamma'(v)\gamma'(u) + \gamma'(u)\gamma'(v) = 2\,g(v,u),
$$
then there exists a unique homomorphism $\phi : \mathcal{C}\ell_{p,q} \to \mathcal{Y}$ such that $\gamma' = \phi \circ \gamma$.
\end{enumerate}
\end{defTEMAi}
\noindent For convenience, the explicit Clifford map is usually suppressed, and the Clifford product is represented by juxtaposition. Accordingly, the algebra elements $\gamma(e_\mu)$ and $\gamma(e_\mu)\gamma(e_\nu)$ are denoted by $\gamma_\mu$ and $\gamma_{\mu\nu}$, respectively, where $\{e_\mu\}_{\mu=0}^{n}$ are basis elements of $\mathbb{R}^{p,q}$. Analogously, $\gamma(e^\mu)$ and $\gamma(e^\mu)\gamma(e^\nu)$ are denoted by $\gamma^\mu$ and $\gamma^{\mu\nu}$, respectively, with $\{e^\mu\}_{\mu=0}^{n}$ a basis of $({\mathbb{R}^\star})^{p,q}$. Here $({\mathbb{R}^\star})^{p,q}$ denotes the dual of the vector space $\mathbb{R}^{p,q}$.

Within the structure of Clifford algebras, in addition to the classical definition of spinors as elements of the space carrying an irreducible representation of the Spin group (the Lorentz group in the case of Minkowski spacetime), there exists an important, albeit less common, algebraic definition~\cite{chev}. Algebraic spinors are minimal left ideals constructed from primitive idempotents of the underlying algebra \cite{Vaz:2016qyw,chev}. Given a Clifford algebra $\mathcal{C}\ell_{p,q}$ and a primitive idempotent $f$, the minimal left ideals take the form $\mathcal{C}\ell_{p,q}f$. Moreover, a division ring $\mathbb{K}$, isomorphic to $\mathbb{R}$, $\mathbb{C}$, or $\mathbb{H}$ (the reals, complexes, or quaternions), is obtained via $f\mathcal{C}\ell_{p,q}f$, depending on the dimension and signature of the vector space.

The mapping
$$
\begin{array}{rcl}
\cdot : \mathcal{C}\ell_{p,q}f \times \mathbb{K} &\longrightarrow& \mathcal{C}\ell_{p,q}f \\[2pt]
(\psi,a) &\longmapsto& \psi \cdot a \equiv \psi a,
\end{array}
$$
defines a right $\mathbb{K}$-module structure on $\mathcal{C}\ell_{p,q}f$. Now we have the enough structure to define algebraic spinor spaces.
\begin{defTEMAi}
The above right $\mathbb{K}$-module $\mathcal{C}\ell_{p,q}f$ is called the \textit{algebraic spinor space} of the algebra $\mathcal{C}\ell_{p,q}$, denoted by $\mathbb{S}_{p,q}$.
Similarly, minimal right ideals $f\mathcal{C}\ell_{p,q}$ can be constructed, giving rise to the space $\mathbb{S}^\star_{p,q}$.
The division ring $\mathbb{K}$ and the module structure are analogous in both cases.
\end{defTEMAi}

\begin{defTEMAi}
An element of $\mathbb{S}^\star_{p,q}$ acts linearly on $\mathbb{S}_{p,q}$, with image in $\mathbb{K}$.
Since $\mathbb{S}^\star_{p,q} \simeq \mathcal{L}(\mathbb{S}_{p,q}, \mathbb{K})$, one may introduce an inner product
$$
\beta: \mathbb{S}_{p,q} \times \mathbb{S}_{p,q} \to \mathbb{K},
\qquad
\beta(\psi, \phi) = \psi^\star \phi,
$$
where $\psi^\star \in \mathbb{S}^\star_{p,q}$ denotes the adjoint (or dual) of $\psi$ with respect to $\beta$.
\end{defTEMAi}

Right ideals can be turned into left ideals, and vice versa, through algebra involutions.
However, idempotents are not necessarily preserved.
If $\alpha$ denotes a generic involution, then $\alpha(\mathcal{C}\ell_{p,q} f) = \alpha(f)\mathcal{C}\ell_{p,q}$, but in general $\alpha(f) \neq f$.
Nevertheless, there always exists an element $h \in \mathcal{C}\ell_{p,q}$ such that $\alpha(f) = h^{-1} f h$ and $\alpha(h) = h$~\cite{Vaz:2016qyw, chev}.
This allows one to define
$$
\psi^\star = \alpha(h\psi) = h\,\alpha(\psi),
$$
leading to the inner product with the respective adjoint involution.

\begin{defTEMAi}
Given an involution $\alpha$ and $h \in \mathcal{C}\ell_{p,q}$ such that $\alpha(f) = h^{-1} f h$ and $\alpha(h) = h$, then
$$
\beta(\psi, \phi) = h\,\alpha(\psi)\,\phi\,f
\in f\mathcal{C}\ell_{p,q}f
\simeq \mathbb{K},
$$
where $\alpha$ is the \textit{adjoint involution} of $\beta$.

\end{defTEMAi}

When dealing with complexified Clifford algebras $\mathbb{C}\!\otimes\!\mathcal{C}\ell_{p,q}$, the composition of complex conjugation with the algebraic involutions modifies the adjoint involution of the inner product.
The relevant situation for us is when the adjoint involution corresponds to Hermitian conjugation in the matrix representation of the algebra. This is achieved whenever
$$
\alpha^*(a) = h^{-1} a^\dagger h,
\qquad h^\dagger = h,
\qquad a \in \mathbb{C}\!\otimes\!\mathcal{C}\ell_{p,q}.
$$
The general adjoint (or dual) spinor $\psi^\star \in \mathbb{C}\!\otimes\!\mathbb{S}^\star_{p,q}$ is then given by
\begin{equation}\label{gdual}
\psi^\star = h\,\alpha^*(\psi) = \psi^\dagger h = [\,h\,\psi\,]^\dagger.
\end{equation}
As a complex associative algebra, $\mathbb{C}\!\otimes\!\mathcal{C}\ell_{p,q}$ is isomorphic to the full matrix algebra $\mathbb{C}\!\otimes\!\mathcal{C}\ell_{p,q} \;\simeq\; M_{4}(\mathbb{C})$. Such isomorphism can be explicitly realized in the Weyl representation of the $\gamma$ matrices\footnote{Throughout this work, we shall adopt the Einstein convention for sum and the standard convention for indices in physics, where Greek indices ($\mu, \nu, \ldots$) range from $0$ to $3$, and Latin indices ($i, j, \ldots$) take values from $1$ to $3$.}
\begin{equation}
\gamma^{\mu}=
\begin{pmatrix}
0 & \bar{\sigma}^{\mu} \\
\sigma^{\mu} & 0
\end{pmatrix},
\qquad
\gamma_{5}=
\begin{pmatrix}
-\,\mathbb{I} & 0 \\
0 & \mathbb{I}
\end{pmatrix},
\qquad
\sigma^{\mu}=(\mathbb{I},\vec{\sigma}),\;\;
\bar{\sigma}^{\mu}=(\mathbb{I},-\vec{\sigma}),
\end{equation}
where $\vec{\sigma} = (\sigma^1,\sigma^2,\sigma^3)$ and $\sigma^i$ denote the Pauli matrices.

Letting $h = \eta\,\Delta$, it can be shown that $\Delta$ does not affect the Lorentz invariance of the inner product, which is ensured by $\eta = \gamma^0$~\cite{Ahluwalia:2016rwl, HoffdaSilva:2019ykt}, highlighting $\Delta$ as the sole source of freedom in defining a Lorentz-covariant dual.

\begin{thmTEMA}
Defining $h = \gamma^0\,\Delta$, the matrix representation of $\Delta$ has the block structure,
$$
\Delta =
\begin{bmatrix}
A & B \\[2pt]
C & A^\dagger
\end{bmatrix},
\quad \text{with} \quad B^\dagger = B, \quad \text{and} \quad C^\dagger = C.
$$
\end{thmTEMA}

\begin{proof}
From $h = \gamma^0\,\Delta$, $h^\dagger = h$ and $\gamma^0 =(\gamma^0)^\dagger$,  follows
\begin{equation}\label{delta}
\Delta^\dagger \gamma^0 = \gamma^0 \Delta.
\end{equation}
Taking a general complex matrix $\Delta = [a_{ij}]$ and imposing~\eqref{delta}, explicit calculations yields
$$
\Delta =
\begin{bmatrix}
a_{11} & a_{12} & a_{13} & a_{14} \\[2pt]
a_{21} & a_{22} & a^*_{14} & a_{24} \\[2pt]
a_{31} & a_{32} & a^*_{11} & a^*_{21} \\[2pt]
a^*_{32} & a_{42} & a^*_{12} & a^*_{22}
\end{bmatrix},
\qquad
a_{13}, a_{31}, a_{24}, a_{42} \in \mathbb{R}.
$$
Defining

$$
A = \begin{bmatrix}
a_{11} & a_{12}  \\[2pt]
a_{21} & a_{22}
\end{bmatrix}, \quad B= \begin{bmatrix}
 a_{13} & a_{14} \\[2pt]
 a^*_{14} & a_{24}
\end{bmatrix}=B^\dagger \quad \text{and} \quad C = \begin{bmatrix}
a_{31} & a_{32} \\[2pt]
a^*_{32} & a_{42}
\end{bmatrix}=C^\dagger,
$$
the block structure $\Delta =
\begin{bmatrix}
	A & B \\[2pt]
	C & A^\dagger
\end{bmatrix}$ follows immediately.\end{proof}

Besides being general, the above matrix form is not particularly convenient for analyzing the influence of the additional parameters. Instead, we aim to preserve the full generality of $\Delta$ while expressing it in terms of the multivector structure naturally inherited from Clifford algebras, following \cite{Rogerio:2025avn}. Regarding the Dirac algebra, the isomorphism $\mathbb{C}\otimes\mathcal{C}\ell_{1,3} \;\simeq\; M_{4}(\mathbb{C})$ guarantees that there must exist a multivector in $\mathbb{C}\otimes \mathcal{C}\ell_{1,3}$ whose matrix representation matches $\Delta$. We denote this multivector as ${A}$. Let us start with a general element of $\mathbb{C}\otimes \mathcal{C}\ell_{1,3}$, denoted by $\bar{A}$. It can be expressed as a linear combination of all possible grades: scalar, vector, bivector, trivector, and pseudoscalar components.
\begin{align}\label{gen_A}
\bar{A} = \alpha
+ a_{\mu}\gamma^{\mu}
+  B_{\mu\nu}\sigma^{\mu\nu}
+ c_{\mu\nu\delta}\gamma^{\mu}\gamma^{\nu}\gamma^{\delta}
+ \beta\,\gamma^{5}, \qquad \delta \neq \mu \neq \nu \neq \delta,
\end{align}
where $\alpha,\beta\in\mathbb{C}$ are scalar coefficients,
$a_{\mu}$ and $c_{\mu\nu\delta}$ are complex four-vectors  and pseudo four-vector components,
$B_{\mu\nu}=-B_{\nu\mu}$ are components of an antisymmetric bivector,
 $\sigma^{\mu\nu}\equiv \frac i2[\gamma^\mu,\gamma^\nu] = \frac i2(\gamma^\mu\gamma^\nu-\gamma^\nu\gamma^\mu)$ form a bivector base and $\gamma^5=\,\gamma^{0}\gamma^{1}\gamma^{2}\gamma^{3}$ is the pseudoscalar.

\begin{thmTEMA}\label{A_real}
Let ${A}$ denote the multivector form corresponding to $\Delta$. In this setting, ${A}$ constitutes a generic invertible element of the Clifford algebra $\mathcal{C}\ell_{1,3}$.
\end{thmTEMA}

\begin{proof}
Analogously to the derivation of the block structure form of $\Delta$ from the condition $\Delta^\dagger \gamma^0 = \gamma^0 \Delta$, the coefficients of $A$ can be determined by taking it as a general element of $\mathbb{C}\otimes\mathcal{C}\ell_{1,3}$ and imposing the constraint $\gamma^{0} A^\dagger \gamma^{0} = A$. In fact, taking $A$ as in eq.~\eqref{gen_A},
\begin{align*}
 {\gamma}^{0}{A}^\dagger{\gamma}^{0} &= {\gamma}^{0}\left( \alpha
+ a_{\mu}\gamma^{\mu}
+  B_{\mu\nu}\sigma^{\mu\nu}
+ c_{\mu\nu\delta}\gamma^{\mu}\gamma^{\nu}\gamma^{\delta}
+ \beta\,\gamma^{5}\right)^\dagger{\gamma}^{0}\\
&= \alpha^* +a_{\mu}^*{\gamma}^{0}{\gamma^{\mu}}^\dagger{\gamma}^{0}
+ B_{\mu\nu}^*{\gamma}^{0}(\sigma^{\mu\nu})^\dagger{\gamma}^{0}
+ c_{\mu\nu\delta}^*{\gamma}^{0}(\gamma^{\mu}\gamma^{\nu}\gamma^{\delta})^\dagger{\gamma}^{0}
+ \beta^*\,{\gamma}^{0}{\gamma^{5}}^\dagger{\gamma}^{0}\\
&= \alpha^* +a_{\mu}^*{\gamma^{\mu}}
+ B_{\mu\nu}^*\sigma^{\mu\nu}
+ c_{\mu\nu\delta}^*\gamma^{\mu}\gamma^{\nu}\gamma^{\delta}
+ \beta^*\,{\gamma^{5}}.
\end{align*}
Thus enforcing $\gamma^{0} A^\dagger \gamma^{0} = A$ constrains all coefficients to be real. Moreover, we observe that $\Delta$ carries 16 real degrees of freedom, and the condition $\det(\Delta)\neq 0$ ensures invertibility without altering this count. Similarly, $A$ has 16 real coefficients and remains fully generic. Under these conditions, $A$ corresponds to a generic invertible element of the Clifford algebra $\mathcal{C}\ell_{1,3}$.
\end{proof}
For the explicit relation between coefficients of $A$ and $\Delta$, see ref. \cite{Rogerio:2025avn}. In Section \ref{cl13}, we explicitly show via a Lie group isomorphism that $A$ has the same 16 degrees of freedom as $\Delta$.

\section{Dual Mappings and Group Structures: The First Attempt}\label{algebras}

We may now explore the dual construction introduced in the previous section by defining a dual mapping $\Omega$, a structure that preserves the full generality of $\Delta$ \cite{Cavalcanti:2020obq}. The main advantage of this approach is that, once defined in this way, certain unexpected algebraic structures within the set of mappings $\Omega$ emerge naturally. We argue that these algebraic structures can be potentially useful in classifying theories of dual spinors. The idea is to fix a starting dual and use the map to connect it to different dual definitions. Our first choice is to start from the dual introduced by Ahluwalia, $\psi^\star = \psi^\dagger\gamma^0\Xi$ \cite{Ahluwalia:2016rwl} (see the Appendix for an explicit form of $\Xi^\dagger$). The $\Omega$ mapping is given by the following definition.
\begin{defTEMAi}
Using the mapping $\Omega$, an arbitrary spinor dual $\psi^\star$ can be defined as
\begin{equation}\label{dualmap}
\psi^\star = [\,\Omega\,\gamma^0\,\Xi\,\psi\,]^\dagger
            = \psi^\dagger\,\gamma^0\,\Xi\,\Omega.
\end{equation}
\end{defTEMAi}
\noindent The relation between $\Delta$ and $\Omega$ follows from comparing Eqs.~\eqref{gdual} with $h=\gamma^0 \Delta$ and~\eqref{dualmap}, giving
$$
\gamma^0\Delta = \Omega\,\gamma^0\,\Xi
\quad \Longleftrightarrow \quad
\Delta = \gamma^0\,\Omega\,\gamma^0\,\Xi,
\quad \text{or} \quad
\Omega = \gamma^0\,\Delta\,\Xi\,\gamma^0.
$$
Using Eq.~\eqref{delta}, one obtains the fundamental restriction on $\Omega$:
\begin{equation}\label{fomega}
\Omega^\dagger
 = \gamma^0\,\Xi^\dagger\,\Delta^\dagger\,\gamma^0
 = \Xi\,\Delta
 = \Xi\,\gamma^0\,\Omega\,\gamma^0\,\Xi.
\end{equation}

We are now in a position to investigate the algebraic structure associated with the mappings $\Omega$. The first question to address is whether a set of such mappings forms a group. To this end, a subset $G_{\Omega} \subset GL(4,\mathbb{C})$ must satisfy associativity, contain the identity element, admit inverses, and be closed under composition. This leads to the following theorem.

\begin{thmTEMA}
A set of $\Omega_i$ forms a group if, and only if, every pair of elements in the set comutates.
\end{thmTEMA}

\begin{proof}
Associativity follows directly from matrix algebra.
The identity corresponds to the case $\Delta = \Xi$.
Invertibility holds since $\det(\Omega)= \det(\Delta) \neq 0$. Moreover, for an invertible $\Omega$ obeying~\eqref{fomega}, one finds
$$
(\Omega^{-1})^\dagger = \Xi\,\gamma^0\,\Omega^{-1}\gamma^0\,\Xi.
$$
Indeed, because $\Omega^\dagger = \Xi\,\gamma^0\,\Omega\,\gamma^0\,\Xi$ and $(\gamma^0)^2 = \Xi^2 = \mathbb{I}$, we have $(\Omega^\dagger)^{-1} = \Xi\,\gamma^0\,\Omega^{-1}\gamma^0\,\Xi$, and thus $(\Omega^{-1})^\dagger = (\Omega^\dagger)^{-1}$.
The closure property is less immediate and imposes a restriction on admissible elements of $G_\Omega$.
Given $\Omega_1$ and $\Omega_2$, Eq.~\eqref{fomega} demands
$$
(\Omega_1\Omega_2)^\dagger
  = \Xi\,\gamma^0\,\Omega_1\,\Omega_2\,\gamma^0\,\Xi.
$$
However,
$$
(\Omega_1\Omega_2)^\dagger
 = \Omega_2^\dagger\,\Omega_1^\dagger
 = \Xi\,\gamma^0\,\Omega_2\,\gamma^0\,\Xi
   \,\Xi\,\gamma^0\,\Omega_1\,\gamma^0\,\Xi
 = \Xi\,\gamma^0\,\Omega_2\,\Omega_1\,\gamma^0\,\Xi.
$$
Comparing both expressions implies
$$
\Omega_1\,\Omega_2 = \Omega_2\,\Omega_1,
$$
meaning that $G_\Omega$ must be an Abelian subgroup of $GL(4,\mathbb{C})$.
\end{proof}

Regarding $\Delta$, the corresponding restriction is
$$
\Delta_1\,\Xi\,\Delta_2\,\Xi = \Delta_2\,\Xi\,\Delta_1\,\Xi.
$$
Since $\Xi^2 = \mathbb{I}$, this reduces to $\Delta_1\,\Xi\,\Delta_2 = \Delta_2\,\Xi\,\Delta_1$. This approach does not allow one to find a fully general description of $G_\Omega$. However, some specific illustrative cases are explored in the following subsection.

The key point to be emphasized here is that the groups $G_{\Omega}$ may be used to classify admissible dual theories. In this context, we introduce an equivalence relation that identifies spinor duals related by the action of $G_{\Omega}$. More precisely, we say that two dual spinors $\psi^\star$ and $\psi^\divideontimes$ are $\Omega$-equivalent, written $\psi^\star \sim_{\Omega} \psi^\divideontimes$, if there exists an element $g_{\Omega} \in G_{\Omega}$ such that
\begin{equation}
    g_{\Omega}\,\psi^\star \;=\; \psi^\divideontimes.
\end{equation}
The corresponding equivalence classes,
\begin{equation}
    [\psi]_{\Omega} \;=\; \{\, g_{\Omega}\,\psi \;\mid\; g_{\Omega} \in G_{\Omega} \,\},
\end{equation}
collect all elements of $\mathbb{S}^\star$ related by a mapping of $G_{\Omega}$. Distinct dual theories are therefore naturally identified with distinct $\Omega-$equivalence classes, providing a sistematic classification of dual spinor structures.

\subsection{Some Particular $G_\Omega$ Groups}

Once the conditions ensuring that $G_\Omega$ forms a group have been established, we may now explore a few explicit particular cases \cite{Cavalcanti:2020obq}. We shall define group elements that combine only $\gamma^0$ and $\Xi$, as seen in Table \ref{tab_elements}. The matrix forms of the elements below are presented in the Appendix.
\begin{table}[h!]
\centering
\renewcommand{\arraystretch}{1.4}
\begin{tabular}{>{$}l<{$} | >{$}l<{$}}
\hline\hline
\text{Element} & \text{Expression} \\ \hline
\mathcal{G}(\phi)  &
\dfrac{m}{2E}\{\gamma^0, \Xi\}
\\[4pt]

\mathcal{F}(\theta,\phi)  &
\dfrac{m}{2p}[\gamma^0, \Xi]
 \\[4pt]

\mathcal{F}(\theta,\phi)\mathcal{G}(\phi)
  &
\dfrac{m^2}{4Ep}[\Xi^\dagger, \Xi] \\[4pt]

\Xi^\dagger(p^\mu)  &
\gamma^0 \Xi \gamma^0 \\[4pt]

\mathcal{G}(\phi)\,\Xi^\dagger(p^\mu) &
\dfrac{m}{2E}(\Xi^\dagger\Xi + \mathbb{I})\gamma^0 \\[4pt]

\mathcal{H}(p^\mu) &
m^2 \Xi \Xi^\dagger \\[4pt]

\mathcal{H}^{-1}(p^\mu) &
m^{-2}\Xi^\dagger \Xi
 = m^{-4}\gamma^0 \mathcal{H} \gamma^0 \\
\hline\hline
\end{tabular}
\caption{Definitions of elements used to construct examples of $G_\Omega$ groups. For the explicit form, see the Appendix.}\label{tab_elements}
\end{table}
From the elements above, three group structures can be identified. Two of them are explicitly defined by
$$
G_\mathcal{F} \equiv
\{\mathbb{I}, \mathcal{G}, \mathcal{F}, \mathcal{F}\mathcal{G}\},
\qquad
G_{\Xi^\dagger} \equiv
\{\mathbb{I}, \mathcal{G}, \Xi^\dagger, \mathcal{G}\Xi^\dagger\},
$$
whose Cayley tables are displayed in Table \ref{tab_G}.
\begin{table}[ht]
\centering
\begin{tabular}{c|cccc}
\hline
$G_\mathcal{F}$ & $\mathbb{I}$ & $\mathcal{G}$ & $\mathcal{F}$ & $\mathcal{F}\mathcal{G}$ \\
\hline
$\mathbb{I}$ & $\mathbb{I}$ & $\mathcal{G}$ & $\mathcal{F}$ & $\mathcal{F}\mathcal{G}$ \\
$\mathcal{G}$ & $\mathcal{G}$ & $\mathbb{I}$ & $\mathcal{F}\mathcal{G}$ & $\mathcal{F}$ \\
$\mathcal{F}$ & $\mathcal{F}$ & $\mathcal{F}\mathcal{G}$ & $\mathbb{I}$ & $\mathcal{G}$ \\
$\mathcal{F}\mathcal{G}$ & $\mathcal{F}\mathcal{G}$ & $\mathcal{F}$ & $\mathcal{G}$ & $\mathbb{I}$ \\
\hline
\end{tabular}
\qquad and \qquad
\begin{tabular}{c|cccc}
\hline
$G_{\Xi^\dagger}$ & $\mathbb{I}$ & $\mathcal{G}$ & $\Xi^\dagger$ & $\Xi^\dagger\mathcal{G}$ \\
\hline
$\mathbb{I}$ & $\mathbb{I}$ & $\mathcal{G}$ & $\Xi^\dagger$ & $\Xi^\dagger\mathcal{G}$ \\
$\mathcal{G}$ & $\mathcal{G}$ & $\mathbb{I}$ & $\Xi^\dagger\mathcal{G}$ & $\Xi^\dagger$ \\
$\Xi^\dagger$ & $\Xi^\dagger$ & $\Xi^\dagger\mathcal{G}$ & $\mathbb{I}$ & $\mathcal{G}$ \\
$\Xi^\dagger\mathcal{G}$ & $\Xi^\dagger\mathcal{G}$ & $\Xi^\dagger$ & $\mathcal{G}$ & $\mathbb{I}$ \\
\hline
\end{tabular}
\caption{Cayley tables for
$G_\mathcal{F}\equiv\{\mathbb{I},\mathcal{G},\mathcal{F},\mathcal{F}\mathcal{G}\}$ (left)
and
$G_{\Xi^\dagger}\equiv\{\mathbb{I},\mathcal{G},\Xi^\dagger,\mathcal{G}\Xi^\dagger\}$ (right).}\label{tab_G}
\end{table}

Both groups are isomorphic to the classical Klein group $K_4$. Despite the isomorphism, $G_\mathcal{F}$ and $G_{\Xi^\dagger}$ are topologically inequivalent, since the parameters defining $G_\mathcal{F}$ are all compact. In contrast, the other group, denoted by $G_\mathcal{H}$, is not of finite order. It is generated by ${\mathbb{I}, \mathcal{F}, \mathcal{G}, \mathcal{H}, \mathcal{H}^{-1}}$ and contains $G_\mathcal{F}$ as a subgroup.

\section{Group Structures Within $\mathcal{C}\ell_{1,3}$}\label{cl13}

In the present section, we analyze the possible $G_{\Omega}$ groups using the multivector structure of Clifford algebras, as introduced at the end of Section~\ref{formulacao}. Besides providing a more elegant formulation, this approach offers access to a substantially richer set of tools for investigating the $G_{\Omega}$ groups. These tools are inherited directly from the algebraic and geometric structure encoded in Clifford algebras \cite{Vaz:2016qyw}. We start by establishing a strong, yet straightforward, result in this context: the largest possible $G_{\Omega}$ group. To this end, we replace the dual structure proposed by Ahluwalia \cite{Ahluwalia:2016rwl}, $\psi^\star = \psi^\dagger \gamma^0 \Xi$, with the standard Dirac dual $\psi^\star = \psi^\dagger \gamma^0$. In this way, the $\Delta$-operator approach becomes formally equivalent to the $\Omega$-mapping approach, since $\Omega = \gamma^0\Delta\gamma^0 = \Delta^\dagger$.

\begin{thmTEMA}
 With respect to the general dual of the Dirac algebra $\mathbb{C}\!\otimes\!\mathcal{C}\ell_{1,3}$, the largest admissible $G_{\Omega}$ group is, up to isomorphism, the group $GL(2,\mathbb{H})$.
\end{thmTEMA}

\begin{proof}
According to Theorem~\ref{A_real}, the multivector representation of the general $\Delta$ (or $\Omega$, as chosen in this section) constitutes a generic element of the algebra $\mathcal{C}\ell_{1,3}$. When restricting to the multiplicative structure, namely, the invertible elements under the Clifford product, the corresponding group is the group of units of the algebra, denoted by $\mathcal{C}\ell_{1,3}^{\times} = \{\omega \in \mathcal{C}\ell_{1,3} \;|\; \exists\; \omega^{-1} \in \mathcal{C}\ell_{1,3} \}$, which is the largest group in $\mathcal{C}\ell_{1,3}$. In view of the isomorphism $\mathcal{C}\ell_{1,3} \simeq M_{2}(\mathbb{H})$ \cite{Vaz:2016qyw}, the group $\mathcal{C}\ell_{1,3}^{\times}$ is represented, at the matrix level, by the group of invertible $2\times 2$ quaternionic matrices, that is, $GL(2,\mathbb{H})$.
\end{proof}

Each quaternion can be represented as a $2 \times 2$ complex matrix via the standard embedding $\mathbb{H} \hookrightarrow M_2(\mathbb{C})$, which induces a natural embedding of $GL(2,\mathbb{H})$ into $GL(4,\mathbb{C})$. Concretely, in a quaternionic matrix
$A = \begin{pmatrix} q_{11} & q_{12} \\ q_{21} & q_{22} \end{pmatrix}$
each $q \in \mathbb{H}$ is mapped to $M_2(\mathbb{C})$ by
$q = a + b\,i + c\,j + d\,k \;\mapsto\; \begin{pmatrix} a+ib & c+id \\ -c+id & a-ib \end{pmatrix}$, explicitly taking the form
\begin{equation}
 A = \begin{pmatrix} a_{11} & a_{12} & a_{13} & a_{14} \\
 -a_{12}^* & a_{11}^* & -a_{14}^* & a_{13}^*\\
 a_{31} & a_{32} & a_{33} & a_{34}\\
-a_{32}^* & a_{31}^* & -a_{34}^* & a_{33}^*
 \end{pmatrix}.
\end{equation}
Notice that the group $GL(2,\mathbb{H})$, and consequently the multivector $A$, has 16 degress of freedom. Consistent with the freedom of $\Delta$ or $\Omega$.

\begin{coroTEMAi}
Restricting $\Delta$ to the even subalgebra\footnote{It is the algebra of even elements in the multivector structure. Here $\hat{\omega}$ denotes the grade involution of $\omega$, defined in homogeneous elements by $\omega_{[p]} = (-1)^p \omega_{[p]}$ and extended by linearity \cite{Vaz:2016qyw}, where $\omega_{[p]}$ is the multivector of grade $p$.} $\mathcal{C}\ell_{1,3}^{+} = \{\omega \in \mathcal{C}\ell_{1,3} \; | \; \omega = \hat{\omega}\}$, the largest admissible $G_{\Omega}$ group is, up to isomorphism, the group $GL(2,\mathbb{C})$.
\end{coroTEMAi}

\begin{proof}
According to the classification theorems of Clifford algebras \cite{Vaz:2016qyw}, one has $(\mathcal{C}\ell^{+}_{1,3})^\times \simeq \mathcal{C}\ell_{3,0} \simeq M_{2}(\mathbb{C})$. In direct analogy with the preceding result, the group of units of the even subalgebra, $(\mathcal{C}\ell^{+}_{1,3})^{\times}$ is thus identified with the group of invertible $2\times 2$ complex matrices, i.e., $GL(2,\mathbb{C})$.
\end{proof}

From the perspective of group theory, several relevant subgroups emerge inside $(\mathcal{C}\ell^{+}_{1,3})^{\times} \simeq GL(2,\mathbb{C})$. The most important is the spin group $\rm{Spin}^+(1,3) \simeq SL(2,\mathbb{C})$, which provides the double covering of the restricted Lorentz group $SO^+(1,3)$. The relationships among these groups can be conveniently summarized by the following commutative diagram.

$$
\begin{tikzcd}[column sep=large, row sep=large]
 \mathcal{C}\ell_{1,3}^\times \arrow[r, "\simeq"]  & GL(2,\mathbb{H}) \arrow[rd, hook] &\\
(\mathcal{C}\ell_{1,3}^+)^\times \arrow[r, "\simeq"] \arrow[u, hook] & GL(2,\mathbb{C}) \arrow[r, hook] \arrow[u, hook]& GL(4,\mathbb{C}) \\
\operatorname{Spin}^+(1,3) \arrow[u, hook] \arrow[r, "\simeq"]& SL(2,\mathbb{C}) \arrow[ur, hook] \arrow[u, hook]&
\end{tikzcd}
$$
The matrix algebra $M(2,\mathbb{C})$ can be natural embedding into $M(4,\mathbb{C})$ by identifying each element of $M(2,\mathbb{C})$ with a block matrix acting on a doubled complex vector space. Explicitly, this embedding is given by
$$
\iota : M(2,\mathbb{C}) \hookrightarrow M(4,\mathbb{C}), \qquad
A \mapsto \begin{pmatrix}
A & 0 \\
0 & A
\end{pmatrix}.
$$
Consequently, $GL(2,\mathbb{C})$ may be viewed as the subgroup of $GL(4,\mathbb{C})$ consisting of block-diagonal matrices with identical $2\times 2$ blocks. In this broader setting, additional symmetries arise: the Pin group $\rm{Pin}(1,3)$, which incorporates reflections and is the double covering of the orthogonal group $O(1,3)$, as well as the full Clifford–Lipschitz group $\Gamma_{1,3} = \left\{
x \in \mathcal{C}\ell_{1,3}^\times \;\middle|\; x v x^{-1} \in \mathbb{R}^{1,3}, \;\forall v \in \mathbb{R}^{1,3}
\right\},$ form intermediate subgroups in the hierarchy of the Clifford algebras structure,
\begin{equation}
\mathrm{Spin}^+(1,3)  \subset \rm{Pin}(1,3) \subset \Gamma_{1,3} \subset \mathcal{C}\ell_{1,3}^{\times}.
\end{equation}
\noindent All of the groups discussed here may be used to induce different partitions of $\mathbb{S}^\star$, each corresponding to a distinct set of $\Omega-$equivalence classes of dual spinors. The diagram below summarizes the structural relationships among the groups discussed above. Solid lines represent proper subgroup inclusions, whereas dashed lines correspond to quotient structures.

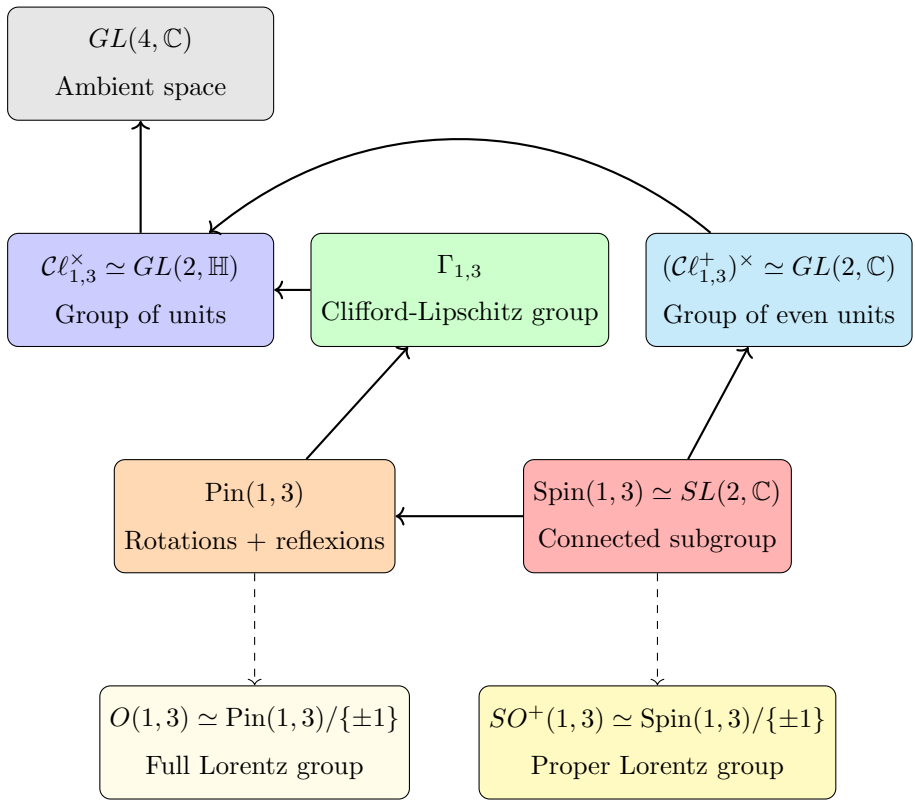
\begin{tikzpicture}[every node/.style={draw, rounded corners, align=center, minimum width=3.5cm, minimum height=1.5cm, font=\small}, column sep=3cm, row sep=2cm]

  \node[fill=gray!20] (GL) at (0,3) {$GL(4,\mathbb{C})$ \\ Ambient space};

  \node[fill=blue!20] (Cl) at (0,0) {$\mathcal{C}\ell_{1,3}^\times \simeq GL(2,\mathbb{H})$ \\ Group of units};
  \node[fill=green!20] (Gamma) at (4.2,0) {$\Gamma_{1,3}$ \\ Clifford-Lipschitz group};
  \node[fill=cyan!20] (ClPlus) at (8.4,0) {$(\mathcal{C}\ell_{1,3}^+)^\times \simeq GL(2,\mathbb{C})$ \\ Group of even units};

  \node[fill=orange!30] (Pin) at (1.5,-3) {$\operatorname{Pin}(1,3)$ \\ Rotations + reflexions};
  \node[fill=red!30] (Spin) at (6.8,-3) {$\operatorname{Spin}(1,3) \simeq SL(2,\mathbb{C})$ \\ Connected subgroup};

  \node[fill=yellow!10] (O) at (1.5,-6) {$O(1,3) \simeq \operatorname{Pin}(1,3)/\{\pm 1\}$ \\ Full Lorentz group};
  \node[fill=yellow!30] (SO) at (6.8,-6) {$SO^+(1,3) \simeq \operatorname{Spin}(1,3)/\{\pm 1\}$ \\ Proper Lorentz group};

  \draw[->, thick] (Cl) -- (GL);
  \draw[->, thick] (Spin) -- (Pin);
  \draw[->, thick] (Pin) -- (Gamma);
  \draw[->, thick] (Gamma) -- (Cl);
  \draw[->, thick] (ClPlus) to[bend right=40] (Cl);
  \draw[->, thick] (Spin) to (ClPlus);

  \draw[->, dashed] (Spin) to (SO);
  \draw[->, dashed] (Pin) to (O);

\end{tikzpicture}


\section{Concluding Remarks}\label{conclusao}

In this work, we advanced the broad investigation of spinor duals originally proposed in \cite{Ahluwalia:2016rwl} within the context of Elko spinors. After reviewing recent developments on general duals \cite{HoffdaSilva:2019ykt} and the corresponding dual mapping groups \cite{Cavalcanti:2020obq}, denoted by $G_\Omega$, we explored these generalized duals and their associated groups by explicitly employing the rich algebraic structure inherent to Clifford algebras, as introduced in \cite{Rogerio:2025avn}. This approach, besides being simpler and more elegant, unveils a variety of possible $G_\Omega$ groups. Consequently, it enables a new classification of spinor duals based on the $\Omega$-equivalence classes proposed here.

Having established the concept of $\Omega$-equivalence classes, we are now in a position to interpret the groups introduced in the previous section within this framework. Each candidate $G_{\Omega}$ acts on the space of spinors by generating orbits that correspond to distinct duality sectors, so that two dual theories are identified exactly when their defining spinors belong to the same $\Omega$-equivalence class. Consequently, the various $G_{\Omega}$ groups presented here provide different levels of refinement in this classification scheme: larger groups produce coarser identifications, while smaller ones generate more classes, and therefore may better distinguish between different duals. In this way, all of the introduced groups serve as natural tools for organizing and distinguishing possible dual spinor theories according to their group–theoretic equivalence classes.


\section*{Appendix: Matrix Representation of $G_\Omega$ Elements} \label{matrices}

Here we explicitly depict matrices appearing along the main text:
\begin{small}
\begin{align*}
 \mathcal{G}(\phi)=\left[\begin{array}{cccc}
0 & 0 & 0 & -ie^{-i\phi} \\
0 & 0 & ie^{i\phi} & 0 \\
0 & -ie^{-i\phi} & 0 & 0 \\
ie^{i\phi} & 0 & 0 & 0
\end{array}\right],
\end{align*}
\end{small}

\begin{small}
\begin{align*}
\mathcal{F}(\theta,\phi)=\left[
\begin{array}{cccc}
 0 & 0 & -\sin \theta & e^{-i \phi } \cos \theta \\
 0 & 0 & e^{i \phi } \cos \theta & \sin \theta \\
 \sin \theta & -e^{-i \phi } \cos \theta & 0 & 0 \\
 -e^{i \phi } \cos \theta & -\sin \theta & 0 & 0 \\
\end{array}
\right],
\end{align*}
\end{small}

\begin{footnotesize}
\begin{align*}
\Xi^\dagger=-\frac{i}{m}\left[
\begin{array}{cccc}
 { p \sin \theta } & { e^{-i \phi } (\text{E}-p \cos \theta)} & 0 & 0 \\
 { -e^{i \phi } (\text{E}+p \cos \theta)} & -{ p \sin \theta} & 0 & 0 \\
 0 & 0 & -{ p \sin \theta} & { e^{-i \phi } (\text{E}+p \cos \theta)} \\
 0 & 0 & -{ e^{i \phi } (\text{E}-p \cos \theta)} & { p \sin \theta} \\
\end{array}
\right],
\end{align*}
\end{footnotesize}

\begin{footnotesize}
\begin{align*}
\mathcal{H}=\left[
\begin{array}{cccc}
 \text{E}^2+2 p \cos \theta \text{E}+p^2 & 2 e^{-i \phi } \text{E} p \sin \theta & 0 & 0 \\
 2 e^{i \phi } \text{E} p \sin \theta & \text{E}^2-2 p \cos \theta \text{E}+p^2 & 0 & 0 \\
 0 & 0 & \text{E}^2-2 p \cos \theta \text{E}+p^2 & -2 e^{-i \phi } \text{E} p \sin \theta \\
 0 & 0 & -2 e^{i \phi } \text{E} p \sin \theta & \text{E}^2+2 p \cos \theta \text{E}+p^2 \\
\end{array}
\right],
\end{align*}
\end{footnotesize}

\begin{footnotesize}
\begin{align*}
\mathcal{H}^{-1}=\left[
\begin{array}{cccc}
 {\text{E}^2-2 p \cos \theta \text{E}+p^2} & -2 e^{-i \phi } \text{E} p \sin \theta & 0 & 0 \\
 -2 e^{i \phi } \text{E} p \sin \theta& \text{E}^2+2 p \cos \theta \text{E}+p^2 & 0 & 0 \\
 0 & 0 & \text{E}^2+2 p \cos \theta \text{E}+p^2 & 2 e^{-i \phi } \text{E} p \sin \theta \\
 0 & 0 & 2 e^{i \phi } \text{E} p \sin \theta & \text{E}^2-2 p \cos \theta \text{E}+p^2
\end{array}
\right].
\end{align*}
\end{footnotesize}

\label{sec7}
\subsection*{\textbf{Acknowledgements}}
 RTC thanks the National Council for Scientific and Technological Development - CNPq (Grant No. 401567/2023-0), for partial financial support. JMHS thanks to CNPq (grant No. 307641/2022-8) for financial support.


\end{document}